\documentclass[a4paper,envcountsect,envcountsame,11pt]{llncs}

\usepackage{amsmath}
\usepackage{amssymb,mathrsfs}  % no amsthm
\usepackage{framed}
\usepackage{enumerate}
\usepackage{color}
\usepackage[normalem]{ulem}
\usepackage{graphicx}
\usepackage{verbatim}
\usepackage{hyperref}
\usepackage{geometry}
\newenvironment{protocol}{\vspace{-1em}\begin{framed}
}{
\vspace{-3ex}\end{framed}\vspace{-1em}}

\renewcommand{\H}{{\cal H}}

\newcommand{\ket}[1]{| #1 \rangle}
\newcommand{\bra}[1]{\langle #1 |}
\newcommand{\braket}[2]{\langle #1|#2 \rangle}
\newcommand{\ketbra}[2]{|#1\rangle \!\langle #2|}
\newcommand{\proj}[1]{\ketbra{#1}{#1}}
\newcommand{\zero}{\ket{0}}
\newcommand{\one}{\ket{1}}
\newcommand{\id}{\mathbb{I}}

\newcommand{\set}[1]{\left \{#1\right \}}
\newcommand{\Set}[2]{\left\{ #1 :\, #2\right\}}
\newcommand{\C}{\mathbb{C}}

\newcommand{\N}{\mathbb{N}}

\newcommand{\eps}{\varepsilon}

\newcommand{\PC}{{\sigma}}  %Pauli Corrections
\newcommand{\gh}{\mathit{GH}}    % garden hose complexity
\newcommand{\ghr}{\mathit{GH_{\eps}}}    % randomized garden hose complexity
\newcommand{\ghq}{\mathit{GH^Q}}    % quantum garden hose complexity
\newcommand{\ghqr}{\mathit{GH_{\eps}^Q}}    % quantum garden hose complexity
\newcommand{\PV}{{\sf PV}}
\newcommand{\PVqubit}{\PV_{\text{\rm\tiny \!qubit}}^f}

\newcommand{\pos}{\text{$po\hspace{-0.1ex}s$}}
\newcommand{\comm}[1]{{}}
\definecolor{darkgreen}{rgb}{0,0.6,0}
\newcommand{\accept}{\sf ACCEPT}
\newcommand{\reject}{\sf REJECT}

\newcommand{\Ltwo}{\mathrm{L}_{(2)}}

\begin{document}

%############################################################################
%          TITLE PAGE
%############################################################################

% Main Title
\title{The Garden-Hose Model}
% \subtitle{A New Model of Computation, and Application to\\Position-Based Quantum Cryptography}

\author{
Harry~Buhrman\inst{1,2}
\and
Serge Fehr\inst{1}
\and
Christian~Schaffner\inst{2,1}
and
Florian Speelman\inst{1}
}

\institute{ 
Centrum Wiskunde \& Informatica (CWI), The Netherlands \and
Institute for Logic, Language and Computation (ILLC), University of Amsterdam, The Netherlands
}

\date{\today}

\maketitle

{ \vspace{-5mm} \small
\center
Email: \email{h.buhrman@cwi.nl}, \email{s.fehr@cwi.nl}, \email{c.schaffner@uva.nl}, \email{f.speelman@cwi.nl}\\
}

\begin{abstract}

We define a new model of communication complexity, called the {\em
garden-hose model}. Informally, the garden-hose complexity of a
function $f:\set{0,1}^n \times \set{0,1}^n \to \set{0,1}$ is given
by the minimal number of water pipes that need to be shared between
two parties, Alice and Bob, in order for them to compute
the function $f$ as follows: Alice connects her ends of the pipes in a
way that is determined solely by her input $x \in \set{0,1}^n$ and,
similarly, Bob connects his ends of the pipes in a way that is
determined solely by his input $y \in \set{0,1}^n$. Alice turns
on the water tap that she also connected to one of the pipes. Then,
the water comes out on Alice's or Bob's side depending on the
function value $f(x,y)$.

We prove almost-linear lower bounds on the garden-hose complexity for concrete functions like inner product, majority, and equality, and we show the existence of functions with exponential garden-hose complexity. Furthermore, we show a connection to classical complexity theory by proving that all functions computable in log-space have polynomial garden-hose complexity. 

We consider a {\em randomized} variant of the garden-hose complexity,
where Alice and Bob hold pre-shared randomness, and a {\em quantum}
variant, where Alice and Bob hold pre-shared quantum entanglement, and
we show that the randomized garden-hose complexity is within a
polynomial factor of the deterministic garden-hose
complexity. Examples of (partial) functions are given where the
quantum garden-hose complexity is logarithmic in $n$ while the classical
garden-hose complexity can be lower bounded by $n^c$ for constant $c>0$.

Finally, we show an interesting connection between the garden-hose model and the (in)security of a certain class of {\em quantum position-verification} schemes. 
\end{abstract}

% \keywords{communication complexity; garden-hose model; position-based
%   quantum cryptography}

\pagestyle{plain}
%############################################################################
\section{Introduction}
%############################################################################

\paragraph{\bf The garden-hose model}
On a beautiful sunny day, Alice and Bob relax in their neighboring
gardens. It happens that their two gardens share $s$ water
pipes, labeled by the numbers $1,2,\ldots,s$. Each of these water
pipes has one loose end in Alice's and the other loose end in Bob's
garden. For the fun of it, Alice and Bob play the following
game. Alice uses pieces of hose to locally connect some of the pipe
ends that are in her garden with each other. For example, she might
connect pipe $2$ with pipe $5$, pipe $4$ with pipe $9$, {\em
  etc}. Similarly, Bob locally connects some of the pipe ends that are
in his garden; for instance pipe $1$ with pipe $4$, {\em etc}. We
note that no T-pieces (nor more complicated constructions),
which connect two or more pipes to one (or vice versa) are
allowed. Finally, Alice connects a water tap to one of her ends of the
pipes, e.g., to pipe $3$ and she turns on the tap. Alice and Bob
observe which of the two gardens gets sprinkled. It is easy to see
that since Alice and Bob only use simple one-to-one connections, there
is no ``deadlock'' possible and the water will indeed eventually come
out on one of the two sides. Which side it is obviously depends on
the respective local connections.

Now, say that Alice connects her ends of the pipes (and the tap) not
in a {\em fixed} way, but her choice of connections depends on a
private bit string $x \in \set{0,1}^n$; for different strings $x$ and
$x'$, she may connect her ends of the pipes differently. Similarly,
Bob's choice which pipes to connect depends on a private bit string $y
\in \set{0,1}^n$. These strategies then specify a function
$f:\set{0,1}^n \times \set{0,1}^n \to \set{0,1}$ as follows: $f(x,y)$
is defined to be $0$ if, using the connections determined by $x$ and
$y$ respectively, the water ends up on Alice's side, and $f(x,y)$ is
$1$ if the water ends up on Bob's side.

Switching the point of view, we can now take an {\em arbitrary}
Boolean function $f:\set{0,1}^n \times \set{0,1}^n \to \set{0,1}$ and
ask: How can $f$ be computed in the garden-hose
model? How do Alice and Bob have to choose their local connections, and how
many water pipes are necessary for computing $f$ in the garden-hose
model?  We stress that Alice's choice for which pipes to connect may
only depend on $x$ but not on $y$, and vice versa; this is what makes
the above questions non-trivial.

In this paper, we introduce and put forward the notion of {\em
  garden-hose complexity}. For a Boolean function $f:\set{0,1}^n \times
\set{0,1}^n \to \set{0,1}$, the garden-hose complexity $\gh(f)$ of $f$
is defined to be the minimal number $s$ of water pipes needed to compute
$f$ in the garden-hose model. It is not too hard to see that $\gh(f)$
is well defined (and finite) for any function $f:\set{0,1}^n \times
\set{0,1}^n \to \set{0,1}$. 

This new complexity notion opens up a large spectrum of natural
problems and questions.  What is the (asymptotic or exact) garden-hose
complexity of natural functions, like equality, inner product {\em
  etc}.? How hard is it to compute the garden-hose complexity in
general? How is the garden-hose complexity related to other complexity
measures? What is the impact of randomness, or entanglement?
Some of these questions we answer in this work; others
remain open.

\paragraph{\bf Lower and upper bounds}

We show a near-linear $\Omega(n/\log(n))$ lower bound on the
garden-hose complexity $\gh(f)$ for a natural class of functions
$f:\set{0,1}^n \times \set{0,1}^n \to \set{0,1}$. This class of
functions includes the mod-$2$ inner-product function, the equality
function, and the majority function. For the former two, this bound is
rather tight, in that for these two functions we also show a linear
upper bound. For the majority function, the best upper bound we know
is quadratic.  Recently, Margalit and Matsliah improved our upper
bound for the equality function with the help of the IBM
SAT-Solver~\cite{MM12} to approximately $1.448 n$, and the question of
how many water pipes are necessary to compute the equality function in
the garden-hose model featured as April 2012's ``Ponder This'' puzzle
on the IBM website\footnote{ \url{http://ibm.co/I7yvMz}}. The {\em
  exact} garden-hose complexity of the equality function is still
unknown, though; let alone of other functions.

By using a counting argument, we show the existence of functions with
{\em exponential} garden-hose complexity, but so far, no such 
function is known explicitly.

\paragraph{\bf Connections to other complexity notions}

We show that every function $f:\set{0,1}^n \times \set{0,1}^n \to \set{0,1}$ that is {\em log-space computable} has polynomial garden-hose complexity. And, vice versa, we show that every function with polynomial garden-hose complexity is, up to local pre-processing, log-space computable. 
As a consequence, we obtain that the set of functions with polynomial garden-hose complexity is exactly given by the functions that can be computed by arbitrary local pre-processing followed by a log-space computation. 

We also point out a connection to communication complexity by observing that, for any function $f:\set{0,1}^n \times \set{0,1}^n \to \set{0,1}$, the {\em one-way communication complexity} of $f$ is a lower bound on $\gh(f) \log(\gh(f))$.

\paragraph{\bf Randomized and quantum garden-hose complexity}

We consider the following natural variants of the garden-hose
model. In the {\em randomized} garden-hose model, Alice and Bob
additionally share a uniformly random string $r$, and the water is
allowed to come out on the wrong side with small probability $\eps$.
Similarly, in the {\em quantum} garden-hose model, Alice and Bob
additionally hold an arbitrary entangled quantum state and their
wiring strategies can depend on the outcomes of measuring this state
before playing the garden-hose game. Again, the water is allowed to
come out on the wrong side with small probability $\eps$. 
Based on the observed connections of the garden-hose complexity to log-space computation and to one-way communication complexity, we can show that the resulting notion of {\em randomized} garden-hose complexity $\ghr(f)$ is polynomially related to $\gh(f)$. For the resulting notion of {\em quantum} garden-hose complexity
$\ghqr(f)$, we can show a separation (for a partial function) from $\ghr(f)$.

\paragraph{\bf Application to quantum position-verification}

Finally, we show an interesting connection between the garden-hose model and the (in)security of a certain class of {\em quantum position-verification} schemes. The goal of position-verification is to verify the geographical position $\pos$ of a prover $P$ by means of sending messages to $P$ and measuring the time it takes $P$ to reply. Position-verification with security against collusion attacks, where different attacking parties collaborate in order to try to fool the verifiers, was shown to be impossible in the classical setting by~\cite{CGMO09}, and in the quantum setting by~\cite{Buhrman2011}, if there is no restriction put upon the attackers. In the quantum setting, this raises the question whether there exist schemes that are secure in case the attackers' quantum capabilities are limited. 

We consider a simple and natural class of
quantum position-verification schemes; each scheme $\PVqubit$ in the class is specified by a Boolean function
$f:\set{0,1}^n \times \set{0,1}^n \to \set{0,1}$. These schemes may
have the desirable property that the more classical resources the
honest users use to faithfully execute the scheme, the more quantum
resources the adversary needs in order to break it. It turns out that there is a one-to-one correspondence between the garden-hose game and a certain class of attacks on these schemes, where the attackers teleport a qubit back and forth using a supply of EPR pairs. As an immediate consequence, the (quantum) garden-hose complexity of $f$ gives an upper bound on
the number of EPR pairs the attackers need in order to break the
scheme $\PVqubit$.  As a corollary, we obtain the following
interesting connection between proving the security of quantum
protocols and classical complexity theory: If there is an $f$ in
$\mathrm{P}$ such that there is no way of attacking scheme $\PVqubit$
using a polynomial number of EPR pairs, then $\mathrm{P} \neq
\mathrm{L}$.  Vice versa, our approach may lead to practical
secure quantum position-verification schemes whose security is based
on classical complexity-theoretical assumptions such as $\mathrm{P}$
is different from $\mathrm{L}$.  However, so far it is still
unclear whether the garden-hose complexity by any means gives a {\em
  lower bound} on the number of EPR pairs needed; this remains to be
further investigated.

%############################################################################
\section{The Garden-Hose Model} \label{sec:gardenhose}
%############################################################################
\subsection{Definition}
Alice and Bob get $n$-bit input strings $x$ and $y$, respectively.
Their goal is to ``compute'' an agreed-upon Boolean function $f:
\set{0,1}^n \times \set{0,1}^n \to \set{0,1}$ on these inputs, in the
following way. Alice and Bob have $s$ water pipes between
them, and, depending on their respective classical inputs $x$ and $y$, they
connect (some of) their ends of the pipes with pieces of hose.  
Additionally, Alice connects a water tap to one of the pipes. They succeed in computing $f$ in the garden-hose model, if the water comes out on Alice's side whenever $f(x,y) = 0$, and the water comes out on Bob's side whenever $f(x,y) = 1$.  Note that it
does not matter out of which pipe the water flows, only on which side
it flows.  What makes the game non-trivial is
that Alice and Bob must do their ``plumbing'' based on their local
input only, and they are not allowed to communicate. We refer to
Figure~\ref{fig:xor} 
for an illustration of computing the XOR function
in the garden-hose model.

\begin{figure}
\center
\includegraphics[width=0.4\textwidth]{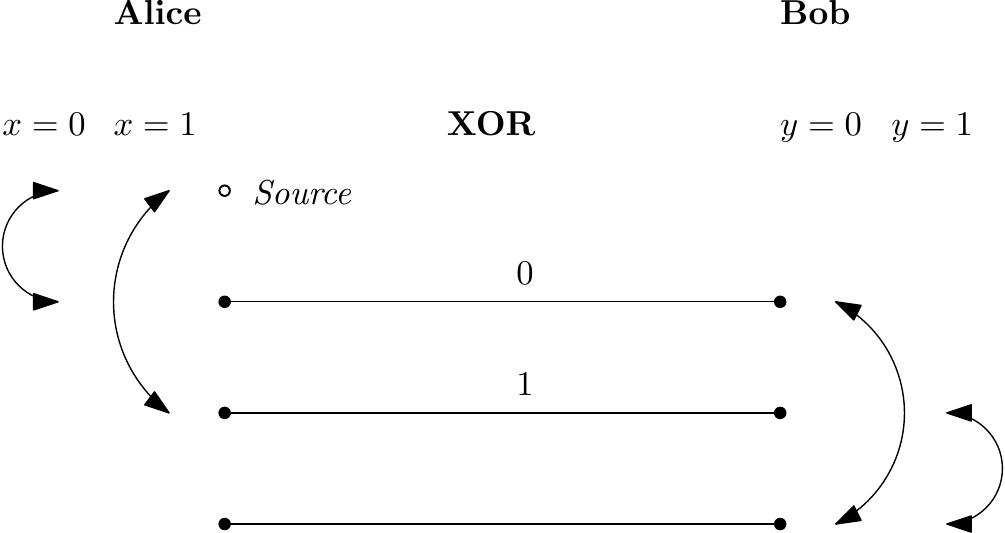}
\caption{Computing the XOR function in the garden-hose model using
  three water pipes. If Alice's input bit $x$ is $0$, she connects the
  water tap to the first water pipe labeled ``0''. In case $x=1$, she
  connects the tap to the second pipe labeled ``1''.  \label{fig:xor}}
\end{figure}

We formalize the above description of the garden-hose game, given in
terms of pipes and hoses {\em etc}., by means of rigorous graph-theoretic
terminology.  However, we feel that the above terminology captures the
notion of a garden-hose game very well, and thus we sometimes use the
above ``watery'' terminology.  We start with a balanced bi-partite
graph $(A\cup B,E)$ which is 1-regular and where the cardinality of
$A$ and $B$ is $|A|=|B|=s$, for an arbitrary large $s \in \N$.  We
slightly abuse notation and denote both the vertices in $A$ and in $B$
by the integers $1,\ldots,s$.  If we need to distinguish $i \in A$
from $i \in B$, we use the notation $i^A$ and $i^B$.  We may assume
that $E$ consists of the edges that connect $i \in A$ with $i \in B$
for every $i \in \set{1,\ldots,s}$, i.e., $E = \Set{\set{i^A,i^B}}{1
  \leq i \leq s}$.  These edges in $E$ are the {\em pipes} in the
above terminology.  We now extend the graph to $(A_\circ \cup B,E)$ by
adding a vertex $0$ to $A$, resulting in $A_\circ = A \cup \set{0}$.
This vertex corresponds to the {\em water tap}, which Alice can
connect to one of the pipes.  Given a Boolean function $f: \set{0,1}^n
\times \set{0,1}^n \to \set{0,1}$, consider two functions
$E_{A_\circ}$ and $E_B$; both take as input a string in $\set{0,1}^n$
and output a set of edges (without self loops). For any $x,y \in
\set{0,1}^n$, $E_{A_\circ}(x)$ is a set of edges on the vertices
$A_\circ$ and $E_B(x)$ is a set of edges on the vertices $B$, so that
the resulting graphs $(A_\circ,E_{A_\circ}(x))$ and $(B,E_B(y))$ have
maximum degree at most $1$.  $E_{A_\circ}(x)$ consists of the {\em
  connections} among the pipes (and the tap) on Alice's side (on input
$x$), and correspondingly for $E_B(y)$.  For any $x,y \in
\set{0,1}^n$, we define the graph $G(x,y) = (A_\circ \cup B,E \cup
E_{A_\circ}(x) \cup E_B(y))$ by adding the edges $E_{A_\circ}(x)$ and
$E_B(y)$ to $E$.  $G(x,y)$ consists of the pipes with the connections
added by Alice and Bob.
Note that the vertex $0 \in A_\circ$ has degree at most $1$, and the
graph $G(x,y)$ has maximum degree at most two $2$; it follows that the
maximal path $\pi(x,y)$ that starts at the vertex $0 \in A_\circ$ is
uniquely determined. $\pi(x,y)$ represents the flow of the water, and
the endpoint of $\pi(x,y)$ determines whether the water comes out on Alice
or on Bob's side (depending on whether the final vertex is in $A_\circ$ or in $B$).

\begin{definition}
  A \emph{garden-hose game} is given by a graph function $G: (x,y)
  \mapsto G(x,y)$ as described above.  The number of pipes $s$ is
  called the \emph{size} of $G$, and is denoted as $s(G)$.  A
  garden-hose game $G$ is said to \emph{compute} a Boolean function $f:
  \set{0,1}^n \times \set{0,1}^n \to \set{0,1}$ if the endpoint of the
  maximal path $\pi(x,y)$ starting at $0$ is in $A_\circ$ whenever
  $f(x,y) = 0$ and in $B$ whenever $f(x,y) = 1$.
\end{definition}

\begin{definition}
The deterministic \emph{garden-hose complexity} of a Boolean function
$f: \set{0,1}^n \times \set{0,1}^n \to \set{0,1}$ is the size $s(G)$
of the smallest garden-hose game $G$ that computes $f$. We denote it
by $\gh(f)$. 
\end{definition}

\subsection{Upper and Lower Bounds} \label{sec:bounds} 
In this section, we present upper and lower bounds on the number of
pipes required to compute some particular (classes of)
functions in the garden-hose model. 
We first give a simple upper bound on $\gh(f)$ which is implicitly
proven in the attack on Scheme II in~\cite{KMS11}. 
\begin{proposition} \label{prop:exp}
  For every Boolean function $f: \set{0,1}^n \times \set{0,1}^n \to
  \set{0,1}$, the garden-hose complexity $\gh(f)$ is at most $2^{n}+1$.
\end{proposition}

\begin{proof}
  We identify $\set{0,1}^n$ with $\set{1,\ldots,2^n}$ in the natural
  way. For $s = 2^n+1$ and the resulting bipartite graph $(A_\circ
  \cup B,E)$, we can define $E_{A_\circ}$ and $E_B$ as
  follows. $E_{A_\circ}(x)$ is set to $\set{(0,x)}$, meaning that Alice connects the tap with the pipe labeled by her input $x$. To define $E_B$,
  group the set $Z(y) = \Set{a \in \set{0,1}^n}{f(a,y)=0}$ arbitrarily
  into disjoint pairs $\set{a_1,a_2} \cup \set{a_3,a_4} \cup \ldots
  \cup \set{a_{\ell-1},a_\ell}$ and set $E_B(y) =
  \set{\set{a_1,a_2},\set{a_3,a_4}, \ldots, \set{a_{\ell-1},a_\ell}}$.
  If $\ell = |Z(y)|$ is odd so that the decomposition into pairs
  results in a left-over $\set{a_\ell}$, then $a_\ell$ is connected
  with the ``reserve'' pipe labeled by $2^n+1$.

By construction, if $x \in Z(y)$ then $x = a_i$ for some $i$, and thus
pipe $x = a_i$ is connected on Bob's side with pipe $a_{i-1}$ or
$a_{i+1}$, depending on the parity of $i$, or with the ``reserve''
pipe, and thus $\pi(x,y)$ is of the form $\pi(x,y) =
(0,x^A,x^B,v^B,v^A)$, ending in $A_\circ$.  On the other hand, if $x
\not\in Z(y)$, then pipe $x$ is not connected on Bob's side, and thus $\pi(x,y) =
(0,x^A,x^B)$, ending in $B$. This proves the claim.  
\qed \end{proof}
We notice that we can extend this proof to show that the garden-hose complexity
$\gh(f)$ is at most $2^{D(f)+1} - 1$, where $D(f)$ is the
deterministic communication complexity
of $f$. See Appendix~\ref{sec:ccupperbound} for a sketch of the method.

\begin{definition}\label{def:injective} We call a function $f$ \emph{injective for Alice}, if for every two different inputs $x$ and $x'$ there exists $y$ such
  that $f(x,y) \neq f(x',y)$. We define \emph{injective for Bob} in an
  analogous way: for every $y \neq y'$, there exists $x$ such that
  $f(x,y) \neq f(x,y')$ holds.
\end{definition}

\begin{proposition}
  If $f$ is injective for Bob or $f$ is injective for Alice,
  then\footnote{All logarithms in this paper are with respect to
  base 2.}  $$\gh(f) \log(\gh(f)) \geq n \, .$$
\end{proposition}

\begin{proof}
  We give the proof when $f$ is injective for Bob. The proof for the
  case where $f$ is injective for Alice is the same.  Consider a
  garden-hose game $G$ that computes $f$. Let $s$ be its size
  $s(G)$. Since, on Bob's side, every pipe is connected to at most one other pipe, there are at most $s^s = 2^{s \log(s)}$
  possible choices for $E_B(y)$, i.e., the set of connections on Bob's side. Thus, if $2^{s \log(s)} < 2^n$, it
  follows from the pigeonhole principle that there must exist $y$ and
  $y'$ in $\set{0,1}^n$ for which $E_B(y) = E_B(y')$, and thus for
  which $G(x,y) = G(x,y')$ for all $x \in \set{0,1}^n$. But this
  cannot be since $G$ computes $f$ and $f(x,y)\neq f(x,y')$ for some
  $x$ due to the injectivity for Bob. Thus, $2^{s \log(s)} \geq 2^n$
  which implies the claim.  \qed
\end{proof}

We can use this result to obtain an almost linear lower bound for several functions that are often
studied in communication complexity settings such as:
\begin{itemize}
\item Bitwise inner product: $\mathrm{IP}(x,y)=\sum_i x_i y_i \pmod{2}$
\item Equality: $\mathrm{EQ}(x,y)= 1$ if and only if $x=y$
\item Majority: $\mathrm{MAJ}(x,y)=1$ if and only if $\sum_i x_i y_i \geq \lceil \frac{n}{2} \rceil$
\end{itemize}
The first two of these functions are injective for both Alice and Bob,
while majority is injective for inputs of Hamming weight at least
$n/2$, giving us the following corollary.
\begin{corollary}
The functions bitwise inner product, equality and majority have
garden-hose complexity in $\Omega(\frac{n}{\log (n)})$.
\end{corollary}
By considering the water pipes that actually get wet, one can show a
lower bound of $n$ pipes for equality~\cite{Pietrzak11}.
On the other hand, we can show upper bounds that are linear for the
bitwise inner product and equality, and quadratic in case of
majority. We refer to~\cite{Speelman11} for the proof of the following
proposition.
\begin{proposition} \label{prop:upperbounds}
  In the garden-hose model, the equality function can be computed with
  $3n+1$ pipes, the bitwise inner product with $4n+1$ pipes and majority
  with $(n+2)^2$ pipes.
\end{proposition}

In general, garden-hose protocols can be transformed into (one-way)
communication protocols by Alice sending her connections
$E_{A_\circ}(x)$ to receiver Bob, which will require at most $\gh(f) \log(\gh(f))$
bits of communication. Bob can then locally compute the function by combining
Alice's message with $E_B(y)$ and checking where the water
exits.\footnote{In fact, garden-hose protocols can even be transformed
  into communication protocols in the more restrictive
  \emph{simultaneous-message-passage} model, where Alice and Bob send
  simultaneous messages consisting of their connections
  $E_{A_\circ}(x)$ and $E_{B}(y)$ to the referee who then computes the
  function. The according statements of
  Propositions~\ref{prop:commtogh}, \ref{prop:commtogh_rand} and
  \ref{prop:commtogh_quant} can be derived analogously.} We summarize this observation in the following proposition.
\begin{proposition} \label{prop:commtogh}
Let $D^{1}(f)$ denote the deterministic one-way
communication complexity of $f$. Then, \[D^{1}(f) \leq \gh(f) \log(\gh(f)) \,\text{.}\]
\end{proposition}
As a consequence, lower bounds on the communication complexity carry
over to the garden-hose complexity (up to logarithmic factors).
Notice that this technique will never give lower bounds that are
better than linear, as problem in communication complexity can always
be solved by sending the entire input to the other party. It is an
interesting open problem to show super-linear lower bounds in the
garden-hose model, e.g.~for the majority function.

\begin{proposition}\label{prop:expbound}\ \\
  There exist functions $f:\set{0,1}^n \times \set{0,1}^n \to
  \set{0,1}$ for which $\gh(f)$ is exponential.
\end{proposition}

\begin{proof}
  The existence of functions with an exponential garden-hose
  complexity can be shown by a simple counting argument. There are
  $2^{2^{2n}}$ different functions $f(x,y)$. For a given size $s =
  s(G)$ of $G$, for every $x \in \set{0,1}^n$, there are at most
  $(s+1)^{s+1}$ ways to choose the connections $E_{A_\circ}(x)$ on
  Alice's side, and thus there are at most $((s+1)^{s+1})^{2^n} =
  2^{2^n (s+1) \log(s+1)}$ ways to choose the function
  $E_{A_\circ}$. Similarly for $E_B$, there are at most $2^{2^n s
    \log(s)}$ ways to choose $E_B$. Thus, there are at most $2^{2\cdot
    2^n (s+1) \log(s+1)}$ ways to choose $G$ of size $s$. Clearly, in
  order for every function $f$ to have a $G$ of size $s$ that computes
  it, we need that $2\cdot 2^n (s+1) \log(s+1) \geq 2^{2n}$, and thus
  that $(s+1) \log(s+1) \geq 2^{n-1}$, which means that $s$ must be
  exponential.  \qed
\end{proof}

\subsection{Polynomial Garden-Hose Complexity and Log-Space Computations}

A family of Boolean functions $\set{f_n}_{n \in \N}$ is
\emph{log-space computable} if there exists a deterministic Turing
machine $M$ and a constant $c$, such that for any $n$-bit input $x$,
$M$ outputs the correct output bit $f_n(x)$, and at most $c \cdot
\log{n}$ locations of $M$'s work tapes are ever visited by $M$'s head
during computation.

\begin{definition}
  We define $\Ltwo$, called \emph{logarithmic space with local
    pre-processing}, to be the class of Boolean functions $f(x,y)$ for
  which there exists a Turing machine $M$ and two \emph{arbitrary}
  functions $\alpha(x),\beta(y)$, such that\footnote{For simplicity of
    notation, we give two arguments to the Turing machine whose concatenation is interpreted
    as the input.} $M(\alpha(x),\beta(y)) =
  f(x,y)$ and $M(\alpha(x),\beta(y))$ runs in space logarithmic in the
  size of the original inputs $|x|+|y|$.
\end{definition}

This definition can be extended in a natural way by considering Turing
machines and circuits corresponding to various complexity classes, and
by varying the number of players. For example, a construction as in
Proposition~\ref{prop:exp} and a similar reasoning as in
Proposition~\ref{prop:logspaceinv} below can be used to show that every
Boolean function is contained in $\mathrm{PSPACE}_{(2)}$. As main
result of this section, we show that our newly defined class $\Ltwo$
is equivalent to functions with polynomial garden-hose complexity. We
leave it for future research to study intermediate classes such as
$\mathrm{AC^0_{(2)}}$ which are related to the polynomial hierarchy of
communication complexity~\cite{BFS86}.

\begin{theorem} \label{thm:main} 
  The set of functions $f$ with polynomial garden-hose complexity
  $\gh(f)$ is equal to $\Ltwo$.
\end{theorem}

The two directions of the theorem follow from Theorem~\ref{thm:logspace} and Proposition~\ref{prop:logspaceinv}.
\begin{theorem} \label{thm:logspace}
  If $f: \set{0,1}^n \times \set{0,1}^n \to \set{0,1}$ is log-space computable, 
then $\gh(f)$ is polynomial in $n$. 
\end{theorem}

\begin{proof}[sketch, full proof in Appendix~\ref{proof:thm:logspace}]
\ \\
  Let $M$ be the deterministic log-space Turing machine deciding
  $f(x,y)=0$. Using techniques from~\cite{Lange1998}, $M$ can be made
  \emph{reversible} incurring only a constant loss in space. As $M$ is a
  log-space machine, it has at most polynomially many
  configurations. The idea for the garden-hose strategy is to label
  the pipes with those configurations of the machine $M$ where the
  input head of $M$ ``switches sides'' from the $x$-part of the input
  to the $y$-part or vice versa. Thanks to the reversibility of $M$,
  the players can then use one-to-one connections to wire up
  (depending on their individual inputs) the open ends of the pipes on
  their side, so that eventually the water flow corresponds to $M$'s
  computation of $f(x,y)$. \qed
\end{proof}

In the garden-hose model, we allow Alice and Bob to locally
pre-process their inputs before computing their wiring. Therefore, it
immediately follows from Theorem~\ref{thm:logspace} that any function
$f$ in $\Ltwo$ has polynomial garden-hose complexity, proving one
direction of Theorem~\ref{thm:main}.

We saw in Proposition~\ref{prop:expbound} that there exist functions with large garden-hose complexity. However, a negative implication of Theorem~\ref{thm:logspace} is that proving the existence of a {\em polynomial-time computable} function $f$ with exponential garden-hose complexity is at least as hard
as separating ${\rm L}$ from ${\rm P}$, a long-standing open problem in complexity theory.
\begin{corollary}
If there exists a function $f: \set{0,1}^n \times \set{0,1}^n \to \set{0,1}$ in P that has super-polynomial garden-hose complexity, then {\rm P} $\neq$ {\rm L}.  
\end{corollary}

It remains to prove the other inclusion of Theorem~\ref{thm:main}.
\begin{proposition} \label{prop:logspaceinv}
Let $f: \set{0,1}^n \times \set{0,1}^n \to \set{0,1}$ be a Boolean
function. If $\gh(f)$ is polynomial (in $n$), then $f$ is in
$\Ltwo$. 
\end{proposition}

\begin{proof}
Let $G$ be the garden-hose game that achieves $s(G) = \gh(f)$. We write $s$ for $s(G)$, the number of pipes, and we let $E_{A_\circ}$ and $E_B$ be the underlying edge-picking functions, which on input $x$ and $y$, respectively, output the connections that Alice and Bob apply to the pipes. 
Note that by assumption, $s$ is polynomial. Furthermore, by the restrictions on $E_{A_\circ}$ and $E_B$, on any input, they consist of at most $(s+1)/2$ connections. 

We need to show that $f$ is of the form $f(x,y) = g(\alpha(x),\beta(y))$, where $\alpha$ and $\beta$ are arbitrary functions $\set{0,1}^n \to \set{0,1}^m$, $g: \set{0,1}^m \times \set{0,1}^m \to \set{0,1}$ is log-space computable, and $m$ is polynomial in $n$. We define $\alpha$ and $\beta$ as follows. For any $x,y \in \set{0,1}^n$, $\alpha(x)$ is simply a natural encoding of $E_{A_\circ}(x)$ into $\set{0,1}^m$, and $\beta(y)$ is a natural encoding of $E_B(y)$ into $\set{0,1}^m$. In the hose-terminology we say that $\alpha(x)$ is a binary encoding of the connections of Alice,
and $\beta(y)$ is a binary encoding of the connections of
Bob. Obviously, these encodings can be done with $m$ of polynomial
size. Given these encodings, finding the endpoint of the maximum path
$\pi(x,y)$ starting in $0$ can be done with logarithmic space: at any
point during the computation, the Turing machine only needs to
maintain a pointer to the position of the water and a binary flag to
remember on which side of the input tape the head is. 
Thus, the function $g$ that computes $g(\alpha(x),\beta(y)) = f(x,y)$ is log-space computable in $m$ and thus also in $n$. 
\qed
\end{proof}

\subsection{Randomized Garden-Hose Complexity} \label{sec:randomized}
It is natural to study the setting where Alice and Bob share a common
random string and are allowed to err with some probability
$\eps$. More formally, we let the players' local strategies
$E_{A_\circ}(x,r)$ and $E_{B}(y,r)$ depend on the shared randomness $r$
and write $G_r(x,y)=f(x,y)$ if the resulting garden-hose game
$G_r(x,y)$ computes $f(x,y)$.

\begin{definition}
  Let $r$ be the shared random string. The \emph{randomized
    garden-hose complexity} of a Boolean function $f: \set{0,1}^n
  \times \set{0,1}^n \to \set{0,1}$ is the size $s(G_r)$ of the
  smallest garden-hose game $G_r$ such that $\forall x,y: \;
  \Pr_r[G_r(x,y)=f(x,y)] \geq 1-\eps$. We denote this minimal size by
  $\ghr(f)$.
\end{definition}

In Appendix~\ref{app:amplification}, we show that the error
probability can be made exponentially small by repeating the protocol
a polynomial number of times.
\begin{proposition}\label{prop:amplification}
 Let $f:\set{0,1}^n \times \set{0,1}^n \to \set{0,1}$ be a function
 such that $\ghr(f)$ is polynomial in $n$, with error $\eps \leq
 \frac{1}{2}-n^{-c}$ for a constant $c>0$.
 For every constant $d>0$ there exists a polynomial $q(\cdot)$ such that
 $\gh_{2^{-n^d}}(f) \leq q\big(\ghr(f) \big)$.
\end{proposition}

Using this result, any randomized strategy can be turned into a
deterministic strategy with only a polynomial overhead in the number of
pipes.
\begin{proposition}\label{prop:derandom}
 Let $f:\set{0,1}^n \times \set{0,1}^n \to \set{0,1}$ be a function
 such that $\ghr(f)$ is polynomial in $n$ and $\eps \leq
 \frac{1}{2}-n^c$ for a constant $c>0$.
 Then there exists a polynomial $q(\cdot)$ such that $\gh(f) \leq q\big(\ghr(f)\big)$.
\end{proposition}

\begin{proof}[sketch]
  By Proposition~\ref{prop:amplification} there exists a randomized
  garden-hose protocol $G_r(x,y)$ of size $q(\ghr(f))$ with error
  probability at most $2^{-2n-1}$. The probability for a random string
  $r$ to be wrong for all inputs is at most $2^{2n}\cdot 2^{-2n-1} <
  1$. In particular, there exists a string $\hat{r}$
  which works for every input $(x,y)$.\qed
\end{proof}

Using this Proposition~\ref{prop:derandom}, we conclude that the lower bound from
Proposition~\ref{prop:expbound} carries over to the randomized setting.
\begin{corollary}\ \\
   There exist functions $f:\set{0,1}^n \times \set{0,1}^n \to
  \set{0,1}$ for which $\ghr(f)$ is exponential.
\end{corollary}

With the same reasoning as in Proposition~\ref{prop:commtogh}, we get
that lower bounds on the randomized one-way communication complexity
with public shared randomness carry over to the randomized garden-hose
complexity (up to a logarithmic factor).
\begin{proposition} \label{prop:commtogh_rand} Let
  $R^{1,\mathrm{pub}}_{\eps}(f)$ denote the minimum
  communication cost of a one-way-communication protocol which
  computes $f$ with an error $\eps$ using public shared
  randomness. Then, $R^{1,\mathrm{pub}}_{\eps}(f) \leq
  \ghr(f) \log(\ghr(f))$.
\end{proposition}
For instance, the linear lower bound $R^{pub}_{\eps}(IP) \in
\Omega(n)$ from~\cite{CG88} for the inner-product function
yields $\ghr(IP) \in \Omega(\frac{n}{\log n})$.

\subsection{Quantum Garden-Hose Complexity} \label{sec:quantumgh}
Let us consider the setting where Alice and Bob share an arbitrary
entangled quantum state besides their water pipes. Depending on their
respective inputs $x$ and $y$, they can perform local quantum
measurements on their parts of the entangled state and wire up the
pipes depending on the outcomes of these measurements. We denote the
resulting \emph{quantum garden-hose complexity} with $\ghq(f)$ in the
deterministic case and with $\ghqr(f)$ if errors are allowed.

With the same reasoning as in Proposition~\ref{prop:commtogh}, we get
that lower bounds on the entanglement-assisted one-way communication
complexity carry over to the quantum garden-hose complexity (up to a
logarithmic factor).
\begin{proposition} \label{prop:commtogh_quant} For $\eps \geq 0$, let
  $Q^1_{\eps}(f)$ denote the minimum cost of an
  entanglement-assisted one-way communication protocol which computes
  $f$ with an error $\eps$. Then,
  $Q^1_{\eps}(f) \leq \ghqr(f) \log(\ghqr(f))$.
\end{proposition}
For instance, the lower bound $Q^1_{\eps}(IP) \in \Omega(n)$ which
follows from results in~\cite{CvDNT98} gives $\ghqr(IP) \in
\Omega(n/\log n)$. For the disjointness function, $Q^1_{\eps}(DISJ)
\in \Omega(\sqrt{n})$ from~\cite{Razborov03} implies $\ghqr(DISJ) \in
\Omega(\sqrt{n}/\log n)$.

In Appendix~\ref{app:separations}, we present partial functions which
give a separation between the quantum and classical
garden-hose complexity in the deterministic and in the randomized setting.
\begin{theorem} \label{thm:separations}
There exist partial Boolean functions $f$ and $g$ such that
\begin{enumerate}
\item $\ghq(f) \in O(\log n)$ \, and \,  $\gh(f) \in
  \Omega(\frac{n}{\log n})$,
\item $\ghqr(g) \in O(\log n)$ \, and \, $\ghr(g) \in
  \Omega(\frac{\sqrt{n}}{\log{n}})$.
\end{enumerate}
\end{theorem}

%############################################################################
\section{Application to Position-Based \\Quantum Cryptography} 
\label{sec:application}
%############################################################################
The goal of {\em position-based cryptography} is to use the
geographical position of a party as its only ``credential".  For
example, one would like to send a message to a party at a geographical
position~$\pos$ with the guarantee that the party can decrypt the
message only if he or she is physically present at~$\pos$.  The
general concept of position-based cryptography was introduced by
Chandran, Goyal, Moriarty and Ostrovsky~\cite{CGMO09}.

A central task in position-based cryptography is the problem of {\em
  position-verification}.  We have a {\em prover}~$P$ at position~$\pos$,
  wishing to convince a set of {\em verifiers} $V_0,\ldots,V_k$
(at different points in geographical space) that~$P$ is indeed at that
position~$\pos$. The prover can run an interactive protocol with the
verifiers in order to convince them.  The main technique for such a
protocol is known as distance bounding~\cite{BC93}. In this technique,
a verifier sends a random nonce to~$P$ and measures the time taken for
$P$ to reply back with this value.  Assuming that the speed of
communication is bounded by the speed of light, this technique gives
an upper bound on the distance of~$P$ from the verifier.

The problem of secure position-verification has been studied before in the
field of wireless security, and there have been several proposals
for this task (\cite{BC93,SSW,VN04,B04} \cite{CH05,SP05,ZLFW06,CCS06}).
However, \cite{CGMO09} shows that there exists no
protocol for secure position-verification that offers security in the presence
of {\em multiple colluding} adversaries. In other words, the set of
verifiers cannot distinguish between the case when they are
interacting with an honest prover at~$\pos$ and the case when they
are interacting with multiple colluding dishonest provers, none of
which is at position~$\pos$.

The impossibility result of~\cite{CGMO09} relies heavily on the fact
that an adversary can locally store all information he receives {\em
  and} at the same time share this information with other colluding
adversaries, located elsewhere. Due to the quantum no-cloning theorem, such a
strategy will not work in the quantum setting, which opens the door to
secure protocols that use quantum information. The quantum model was
first studied by Kent et al. under the name of ``quantum
tagging''~\cite{KMSB06,KMS11}. Several schemes were
developed~\cite{KMS11,Mal10a,CFGGO10,Mal10b,LL11} and proven later to
be insecure. Finally in~\cite{Buhrman2011} it was shown that in
general no unconditionally secure quantum position-verification scheme is
possible. Any scheme can be broken using a double exponential amount
of EPR pairs in the size of the messages of the protocol. Later, Beigi
and K\"onig improved in~\cite{Beigi2011} the double exponential
dependence to single exponential making use of port-based
teleportation~\cite{IH08,IH09}.

Due to the exponential overhead in EPR pairs, the general no-go
theorem does not rule out the existence of quantum schemes that are
secure for all practical purposes. Such schemes should have the
property that the protocol, when followed honestly, is feasible, but
cheating the protocol requires unrealistic amounts  of  resources, for
example EPR pairs or time.

\subsection{A Single-Qubit Scheme} \label{sec:motivation} 
Our original motivation for the garden-hose model was to study a
particular quantum protocol for secure position verification,
described in Figure~\ref{fig:PVqubit}. The protocol is of the generic
form described in Section~3.2 of~\cite{Buhrman2011}. In
Step~\ref{step:preparation}, the verifiers prepare challenges for the
prover. In Step~\ref{step:send}, they send the challenges, timed in
such a way that they all arrive at the same time at the prover. In
Step~\ref{step:prover}, the prover computes his answers and sends them
back to the verifiers. Finally, in Step~\ref{step:verification}, the
verifiers verify the timing and correctness of the answer.

As in~\cite{Buhrman2011}, we consider here for simplicity the case
where all players live in one dimension, the basic ideas generalize to
higher dimensions. In one dimension, we can focus on the case of two
verifiers $V_0, V_1$ and an honest prover $P$ in between them.

We minimize the amount of quantum communication in that only one
verifier, say $V_0$, sends a qubit to the prover, whereas both
verifiers send classical $n$-bit strings $x,y \in \set{0,1}^n$ that
arrive at the same time at the prover. We fix a publicly known Boolean
function $f: \set{0,1}^n \times \set{0,1}^n \rightarrow \set{0,1}$
whose output $f(x,y)$ decides whether the prover has to return the
qubit (unchanged) to verifier $V_0$ (in case $f(x,y)=0$) or to
verifier $V_1$ (if $f(x,y)=1$).

\begin{figure}[htb]
\small
\begin{protocol}
\begin{enumerate}\setlength{\parskip}{0.1ex}\setcounter{enumi}{-1}
\item\label{step:preparation} $V_0$ randomly chooses two $n$-bit
  strings $x,y \in \set{0,1}^n$ and privately sends $y$
  to~$V_1$. $V_0$ prepares an EPR pair $(\ket{0}_V\ket{0}_P +
  \ket{1}_V\ket{1}_P)/\sqrt{2}$. If $f(x,y)=0$, $V_0$ keeps the qubit
  in register $V$. Otherwise, $V_0$ sends the qubit in register $V$
  privately to $V_1$.
\item\label{step:send} $V_0$ sends the qubit in register $P$ to the prover $P$ together
  with the classical $n$-bit string $x$. $V_1$ sends $y$ so that it
  arrives at the same time as the information from $V_0$ at $P$.
\item\label{step:prover} $P$ evaluates $f(x,y) \in \set{0,1}$ and
  routes the qubit to $V_{f(x,y)}$.
\item\label{step:verification} $V_0$ and $V_1$ accept if the qubit
  arrives in time at the correct verifier and the Bell measurement of
  the received qubit together with the qubit in $V$ yields the correct
  outcome.
\end{enumerate}
\end{protocol}
 \caption{Position-verification scheme $\PVqubit$ using one
     qubit and classical $n$-bit strings.}
 \label{fig:PVqubit}
\end{figure}

The motivation for considering this protocol is the following: As the
protocol uses only one qubit which needs to be correctly routed, the
honest prover's quantum actions are trivial to perform. His main task
is evaluating a classical Boolean function $f$ on classical inputs $x$
and $y$ whose bit size $n$ can be easily scaled up.  On the other
hand, our results suggest that the adversary's job of successfully
attacking the protocol becomes harder and harder for larger input
strings $x,y$.

\subsection{Connection to the Garden-Hose Model}
In order to analyze the security of the protocol $\PVqubit$, we define
the following communication game in which Alice and Bob play the roles
of the adversarial attackers of $\PVqubit$. Alice starts with an
unknown qubit $\ket{\phi}$ and a classical $n$-bit string $x$ while
Bob holds the $n$-bit string $y$. They also share some quantum state
$\ket{\eta}_{AB}$ and both players know the Boolean function $f:
\set{0,1}^n \times \set{0,1}^n \to\{0,1\}$. The players are allowed
one round of simultaneous classical communication combined with
arbitrary local quantum operations. When $f(x,y)=0$, Alice should be
in possession of the qubit $\ket{\phi}$ at the end of the
protocol and on $f(x,y)=1$, Bob should hold it.

As a simple example consider the case where $f(x,y)=x\oplus y$, the
XOR function, with 1-bit inputs $x$ and $y$. Alice and Bob
then have the following way of performing this task perfectly by using
a pre-shared quantum state consisting of three EPR pairs (three
ebits). Label the first two EPR pairs $0$
and $1$. Alice teleports\footnote{See
  Appendix~\ref{sec:teleportation} for a brief introduction to quantum
  teleportation.} $\ket{\phi}$ to Bob using the pair labeled
with her input $x$. This yields measurement result $i\in\{0,1,2,3\}$,
while Bob teleports his half of the EPR pair labeled $y$ to Alice
using his half of the third EPR pair while obtaining measurement
outcome $j \in \set{0,1,2,3}$ . In the round of simultaneous
communication, both players send the classical measurement results and
their inputs $x$ or $y$ to the other player. If $x\oplus y=1$,
i.e.~$x$ and $y$ are different bits, Bob can apply the Pauli operator
$\sigma_{i}$ to his half of the EPR pair labeled $x=y\oplus 1$,
correctly recovering $\ket{\phi}$. Similarly, if $x\oplus y=0$, it is
easy to check that Alice can recover the qubit by applying
$\sigma_{i}\sigma_{j}$ to her half of the third EPR pair.

If Alice and Bob are {\em constrained} to the types of actions in the
example above, i.e., if they are restricted to teleporting the quantum
state back and forth depending on their classical inputs, there is a one-to-one correspondence between attacking the position-verification scheme $\PVqubit$ and computing the function $f$ in the garden-hose model. The quantum strategy for attacking $\PVqubit$ in the
example above exactly corresponds to the strategy depicted in
Figure~\ref{fig:xor} for computing the XOR-function in the garden-hose
model. 

More generally, we can translate any strategy of Alice and Bob
in the garden-hose model to a perfect quantum attack of $\PVqubit$ by
using one EPR pair per pipe and performing Bell measurements where the
players connect the pipes.

Our hope is that also the converse is true: if many pipes are required
to compute $f$ (say we need super-polynomially many), then the number
of EPR pairs needed for Alice and Bob to successfully break $\PVqubit$
with probability close to $1$ by means of an {\em arbitrary} attack
(not restricted to Bell measurements on EPR pairs) should also be
super-polynomial.

The examples of (partial) functions from Theorem~\ref{thm:separations}
show that the classical garden-hose complexity $\gh(f)$ does not
capture the amount of EPR pairs required to attack $\PVqubit$. It is
conceivable that one can show that arbitrary attacks can be cast in
the quantum garden-hose model and hence, the quantum garden-hose
complexity $\ghqr(f)$ (or a variant of it\footnote{In addition to the
  number of pipes, one might have to account for the size of the
  entangled state as well.})  correctly captures the amount of EPR
pairs required to attack $\PVqubit$. We leave this question as an
interesting problem for future research.

We stress that for this application, any polynomial lower bound on the
number of required EPR pairs is already interesting.

\subsection{Lower Bounds on Quantum Resources to Perfectly Attack
  $\PVqubit$} \label{sec:lowerbound} In Appendix~\ref{app:lowerbound},
we show that for a function that is injective for Alice or injective
for Bob (according to Definition~\ref{def:injective}), the dimension
of the quantum state the adversaries need to handle (including
possible quantum communication between them) in order to attack
protocol $\PVqubit$ perfectly has to be of order at least linear in
the classical input size $n$. In other words, they require at least a
logarithmic number of qubits in order to successfully attack
$\PVqubit$.
\begin{theorem}
  Let $f$ be injective for Bob. Assume that Alice and Bob perform a
  perfect attack on protocol $\PVqubit$. Then, the dimension $d$ of
  the overall state (including the quantum communication) is in
  $\operatorname\Omega(n)$.
\end{theorem}

In the last subsection, we show that there exist functions for which
perfect attacks on $\PVqubit$ requires the adversaries to handle a
polynomial amount of qubits.
\begin{theorem}
  For any starting state $\ket{\psi}$ of dimension $d$, there exists a
  Boolean function $f$ on inputs $x,y \in \set{0,1}^n$ such that any perfect
  attack on $\PVqubit$ requires $d$ to be exponential in $n$.
\end{theorem}

These results can be seen as first steps towards establishing the
desired relation between classical difficulty of honest actions and
quantum difficulty of the actions of dishonest players. We leave as
future work the generalization of these lower bounds to the more
realistic case of imperfect attacks and also to more relevant
quantities like some entanglement measure between the players (instead
of the dimension of their shared state).

%############################################################################
\section{Conclusion and Open Questions}  \label{sec:openquestions}
%############################################################################

The garden-hose model is a new model of communication complexity. We
connected functions with polynomial garden-hose complexity to a newly
defined class of log-space computations with local
pre-processing. Alternatively, the class $\Ltwo$ can also be viewed as
the set of functions which can be decided in the
simultaneous-message-passing (SMP) model where the referee is
restricted to log-space computations. Many open questions remain. Can
we find better upper and lower bounds for the garden-hose complexity
of the studied functions? The constructions given in~\cite{Speelman11}
still leave a polynomial gap between lower and upper bounds for many
functions. It would also be interesting to find an explicit function
for which the garden-hose complexity is provably super-linear or even
exponential, the counting argument in Proposition~\ref{prop:expbound}
only shows the existence of such functions.  It is possible to extend
the basic garden-hose model in various ways and consider settings with
more than two players, non-Boolean functions or multiple water
sources.  Furthermore, it is interesting to relate our findings to
very recent results about space-bounded communication
complexity~\cite{BCPSS12}.

Garden-hose complexity is a tool for the analysis of a specific
scheme for position-based quantum cryptography. This scheme requires
the honest prover to work with only a single qubit, while the
dishonest provers potentially have to manipulate a large quantum
state, making it an appealing scheme to further examine.  The
garden-hose model captures the power of attacks that only use
teleportation, giving upper bounds for the general scheme, and lower
bounds when restricted to these attacks.

An interesting additional restriction on the garden-hose model would
involve limiting the computational power of Alice and Bob. For example
to polynomial time, or to the output of quantum circuits of polynomial
size. Bounding not only the amount of entanglement, but also the
amount of computation with a realistic limit might yield stronger
security guarantees for the cryptographic schemes.

%############################################################################
\section{Acknowledgments}  \label{sec:acknowledgements}
%############################################################################
HB is supported by an NWO Vici grant and the EU project
QCS. FS is supported by the NWO DIAMANT project.
CS is supported by an NWO Veni grant. We thank Louis Salvail for
useful discussions about the protocol $\PVqubit$.

\appendix
%############################################################################
%\section*{Appendices}
%############################################################################
\section{Upper Bound by Communication Complexity}\label{sec:ccupperbound}
We show that the garden-hose complexity
$\gh(f)$ of any function $f$ is at most $2^{D(f)+1} - 1$, where $D(f)$
is the deterministic communication complexity
of $f$.

Consider a protocol where Alice and Bob alternate in sending one bit.
The pipes between Alice and Bob are labeled with all possible
non-empty strings of length up to $D(f)$, with one extra reserve pipe.

Let $A_v(x)$ be the bit Alice sends after seeing transcript $v \in \set{0,1}^*$ given input $x$ and let $B_v(x)$ be
the bit Bob sends after a transcript $v$ on input $y$. (Since Alice and Bob alternate, Alice sends a bit
on even length transcripts, while Bob sends when the transcript has odd length.)
Alice connects the tap to $0$ or $1$ depending on the first sent bit. Then, Alice
makes connections 
\[
\{\{v, v A_v(x)\} | v \in \set{0,1}^* \text{with } |v| \text{ even and } 1 \leq |v| \leq D(f)\} \,\text{.}
\]
Here $v A_v(x)$ is the concatenation of $v$ and $A_v(x)$.
Bob's connections are given by the set
$$\set{\set{v, v B_v(x)} | v \in \set{0,1}^* \text{with } |v| \text{
    odd and } 1 \leq |v| \leq D(f)} \, .$$
Now, for all transcripts of length $D(f)$, Alice knows the function outcome. (Assume $D(f)$ is even for simplicity.)
For those $2^{D(f)}$ pipes she can route the water to the correct side by connecting similar outcomes, as in the proof of
Proposition~\ref{prop:exp}, using one extra reserve pipe. This brings the total used pipes to
$1+\sum_{i=1}^{D(f)} 2^i = 2^{D(f)+1}-1$.
The correctness can be verified by comparing the path of the water to the communication protocol: the label
of the pipe the water is in, when following it through the pipes for $r$ ``steps'', is exactly the same as
the transcript of the communication protocol when executing it for $r$ rounds.

\section{Proofs}
\subsection{Proof of Theorem~\ref{thm:logspace}}
%\begin{theorem*} 
{\sc Theorem~\ref{thm:logspace}} \emph{
  If $f: \set{0,1}^n \times \set{0,1}^n \to \set{0,1}$ is log-space computable, 
then $\gh(f)$ is polynomial in $n$. 
}
%\end{theorem*}
\label{proof:thm:logspace}
\begin{proof}
Let $M$ be a deterministic Turing machine deciding $f(x,y)=0$. We
assume that $M$'s read-only input tape is of length $2n$ and contains
$x$ on positions $1$ to $n$ and $y$ on positions $n+1$ to $2n$. By
assumption $M$ uses logarithmic space on its work tapes.

In this proof, a \emph{configuration} of $M$ is the location of its
tape heads, the state of the Turing machine and the content of its
work tapes, excluding the content of the read-only input tape.  This
is a slightly different definition than usual, where the content of
the input tape is also part of a configuration. When using the normal
definition (which includes the content of all tapes), we will use the term
\emph{total configuration}. Any configuration of $M$ can be described
using a logarithmic number of bits, because $M$ uses logarithmic
space.

A Turing machine is called \emph{deterministic}, if every total
configuration has a unique next one. A Turing machine is called
\emph{reversible} if in addition to being deterministic, every total configuration also has a unique predecessor.
An $S(n)$ space-bounded deterministic
Turing machine can be simulated by a reversible Turing machine in
space $O(S(n))$~\cite{Lange1998}.
This means that without
loss of generality, we can assume $M$ to be a reversible Turing
machine, which is crucial for our construction. Let $M$ also be
\emph{oblivious}\footnote{A Turing machine is called \emph{oblivious},
  if the movement in time of the heads only depend on the length of the input, known in advance to be $2n$, but
  not on the input itself. For our construction we only require the input tape head to have this property.} in the tape head movement on the input tape. This can be done with only a small increase in space by adding a
counter. 

Alice's and Bob's perfect strategies in the garden-hose game are as
follows. They list all configurations where the head of the input tape
is on position $n$ coming from position $n+1$. Let us call the set of
these configurations $C_{A}$. Let $C_{B}$ be the analogous set of
configurations where the input tape head is on position $n+1$ after
having been on position $n$ the previous step. Because $M$ is
oblivious on its input tape, these sets depend only on the function
$f$, but not on the input pair $(x,y)$.  The number of elements of $C_A$
and $C_B$ is at most polynomial, being exponential in the description
length of the configurations.  Now, for every element in $C_{A}$ and
$C_{B}$, the players label a pipe with this configuration.  Also label $|C_{A}|$
pipes $\accept$ and $|C_{B}|$ of them $\reject$.  These steps determine
the number of pipes needed, Alice and Bob can do this labeling
beforehand.

For every configuration in $C_{A}$, with corresponding pipe $p$, Alice
runs the Turing machine starting from that configuration until it
either accepts, rejects, or until the input tape head reaches position
$n+1$. If the Turing machine accepts, Alice connects $p$ to the first free pipe labeled $\accept$.
On a reject, she leaves $p$ unconnected. If
the tape head of the input tape reaches position $n+1$, she connects
$p$ to the pipe from $C_{B}$ corresponding to the configuration of the
Turing machine when that happens. By her knowledge of $x$, Alice knows
the content of the input tape on positions $1$ to $n$, but not the
other half.
Alice also runs $M$ from the starting configuration, connecting the
water tap to a target pipe with a configuration from $C_{B}$
depending on the reached configuration.

Bob connects the pipes labeled
by $C_{B}$ in an analogous way: He runs the Turing machine starting
with the configuration with which the pipe is labeled until it halts
or the position of the input tape head reaches $n$. On accepting, the
pipe is left unconnected and if the
Turing machine rejects, the pipe is connected to one of the pipes labeled $\reject$. Otherwise, the
pipe is connected to the one labeled with the configuration in
$C_{A}$, the configuration the Turing machine is in when the head on
the input tape reached position $n$.

In the garden-hose game, only one-to-one connections of pipes are
allowed. Therefore, to check that the described strategy is a valid one,
the simulations of two different configurations
from $C_A$ should never reach the same configuration in
$C_B$. This is guaranteed by the reversibility of $M$ as follows. 
Consider Alice simulating $M$ starting from different configurations $c
\in C_A$ and $c' \in C_A$. We have to check that their simulation can not end at
the same $d \in C_B$, because Alice can not connect both pipes labeled
$c$ and $c'$ to the same $d$. Because $M$ is reversible, we can in principle also simulate $M$ backwards in time starting
from a certain configuration. In particular, Alice can simulate $M$ backwards starting with configuration $d$,
until the input tape head position reaches $n+1$. The configuration of $M$ at that time can not simultaneously
be $c$ and $c'$, so there will never be two different pipes trying to connect to the pipe labeled $d$.

It remains to show that, after the players link up their pipes as described,
the water comes out on Alice's side if $M$ rejects on input $(x,y)$,
and that otherwise the water exits at Bob's.  We can verify the correctness of the described strategy
by comparing the flow of the water directly to the execution of $M$. Every pipe the
water flows through corresponds to a configuration of $M$ when it runs starting from the
initial state. So the side on which the water finally exits also
corresponds to whether $M$ accepts or rejects.
\qed \end{proof}

%############################################################################
\subsection{Proof of Proposition~\ref{prop:amplification}}
\label{app:amplification}
%############################################################################
%\begin{proposition*}
{\sc Proposition~\ref{prop:amplification}} \emph{
 Let $f:\set{0,1}^n \times \set{0,1}^n \to \set{0,1}$ be a function
 such that $\ghr(f)$ is polynomial in $n$, with error $\eps \leq
 \frac{1}{2}-n^{-c}$ for a constant $c>0$.
 For every constant $d>0$ there exists a polynomial $q(\cdot)$ such that
 $\gh_{2^{-n^d}}(f) \leq q\big(\ghr(f) \big)$.
}
\begin{proof}
The new protocol $G'_r(x,y)$ takes the majority of $k=8n^{2c+d}$
outcomes of $G_{r_i}(x,y)$ where $r_1,\ldots,r_k$ are $k$
independent and uniform samples of the random string. We have to
establish (i) that taking the majority of $k$ instances of the original protocol
indeed gives the correct outcome with probability at least $1-2^{-n^d}$ and (ii) that $G'_r(x,y)$ requires only polynomial pipes.

\begin{list}{(\labelitemi)}{\leftmargin=1.5em}
\item[(i)] Let $X_i$ be the random variable that equals $1$ when
  $G_{r_i}(x,y) = f(x,y)$ and $0$ otherwise. Note that the $X_i$ are
  independent and identically distributed random variables with
  expectation $E[X_i] \geq 1-\eps =: p$.  Whenever $\sum_{i=1}^k X_i \geq
  \frac{k}{2}$ the protocol gives the correct outcome. Use the
  Chernoff bound to get
 \[
  \Pr\left[\sum_{i=1}^k X_i < (1 - \zeta) p k\right] \leq e^{-\frac{\zeta^2}{2} p k}
 \]
for any small $\zeta$. Picking $\zeta=n^{-c}$, so that $(1-\zeta)p k$ is still greater than $\frac{k}{2}$,
and filling in $k$, we can upper bound the probability of failure by
\[
 e^{-\frac{8 n^{2c+d}}{2 n^{2c}}p} \leq 2^{-n^d}
\]

 \item[(ii)]
In Theorem~\ref{thm:logspace} we show that any log-space computable function can be simulated by a polynomial-sized garden-hose strategy.
Thus, if checking the majority of $k$ garden-hose strategies can be done in logarithmic space (after local pre-computations by Alice
and Bob), then $G'_r(x,y)$ can be computed using a polynomial number of pipes.

Let $A_i = E_{A_\circ}(x,r_i)$ be the local wiring of Alice for strategy $G$ on input $x$ with randomness $r_i$, and let $B_i=E_{B}(y,r_i)$.
Alice locally generates $(A_1, \ldots, A_k)$ and Bob locally generates $(B_1, \ldots, B_k)$.
In the proof of Proposition~\ref{prop:logspaceinv} it was shown that simulating the outcome of a single garden-hose strategy $(A_i,B_i)$ can be done in logarithmic space.
Here we follow the same construction, but instead of getting the outcome of a single strategy we simulate all $k$ strategies.
This can still be done in logarithmic space, since we can re-use the memory needed to simulate each of the $k$ strategies.
To find the majority, we need to add a counter to keep track of the simulation outcomes, using only an extra $\log k$ bits of space.
\qed
\end{list}

\end{proof}

%############################################################################
\section{Quantum Preliminaries} \label{sec:qpreliminaries}
%############################################################################
For Appendices~\ref{app:separations} and~\ref{app:lowerbound}, we
assume that the reader is familiar with basic concepts of quantum
information theory. We refer to~\cite{NC00} for an introduction and
merely fix some notation here.

\subsection{Quantum Teleportation} \label{sec:teleportation}
An important example of a $2$-qubit state is the {\em EPR pair}, which
is given by $\ket{\Phi}_{AB} = (\zero_A\zero_B + \one_A\one_B)/\sqrt{2} \in
\H_A \otimes \H_B = \C^2 \otimes \C^2$ and has the following
properties: if qubit $A$ is measured in the computational basis, then
a uniformly random bit $x \in \set{0,1}$ is observed and qubit $B$
collapses to $\ket{x}$. Similarly, if qubit $A$ is measured in the
Hadamard basis, then a uniformly random bit $x \in \set{0,1}$ is
observed and qubit $B$ collapses to $H\ket{x}$.

The goal of quantum teleportation is to transfer a quantum state from
one location to another by only communicating classical information.
Teleportation requires pre-shared entanglement among the two
locations.  To teleport a qubit~$Q$ in an arbitrary unknown
state~$\ket{\psi}_Q$ from Alice to Bob, Alice performs a
Bell-measurement on $Q$ and her half of an EPR pair, yielding a
classical measurement outcome $k \in \set{0,1,2,3}$.  Instantaneously,
the other half of the corresponding EPR pair, which is held by Bob,
turns into the state $\PC_k \ket{\psi}$, where $\PC_0, \PC_1,
\PC_2, \PC_3$ denote the four Pauli-corrections $\set{\id,X,Z,XZ}$,
respectively. The classical information $k$ is then communicated
to Bob who can recover the state $\ket{\psi}$ by performing $\PC_k$ on
his EPR half. %  Note that the Pauli operators $\PC_k$ are Hermitian, thus
% $\PC^\dag_k=\PC_k$.

\section{Separations between Quantum and Classical
  Garden-Hose Complexity} \label{app:separations}

\subsection{Deterministic Setting}
Using techniques from~\cite{BCW98}, we show a separation between the
garden-hose model and the quantum garden-hose model in the
deterministic setting for the function $EQ'$, defined as:
\[
 EQ'(x,y) = \left\{ 
  \begin{array}{l l}
    1 & \quad \mathrm{if} \; \Delta (x,y) = 0 \, , \\
    0 & \quad \mathrm{if} \; \Delta(x,y) = n/2 \, , \\
  \end{array} \right.
\]
where $\Delta(x,y)$ denotes the Hamming distance between two $n$-bit strings $x$ and $y$.
We show that the zero-error quantum garden-hose complexity of $EQ'$ is logarithmic in the input length.

\begin{theorem}
 $\ghq (EQ') \in O(\log n)$.
\end{theorem}

\begin{proof}
  Alice and Bob start with the fully entangled quantum state of $\log
  n$ qubits, i.e.~with
  $\frac{1}{\sqrt{n}}\sum_{i=0}^{n-1}\ket{i}\ket{i}$. Counting indices
  of the input bits from 0 to $n-1$, Alice gives a phase of $-1$ to state
  $\ket{i}$ whenever $x_{i}=0$ and Bob does the same thing with his
  half when the bit $y_{i}=0$, yielding the state
\[
\frac{1}{\sqrt{n}}\sum_{i=0}^{n-1}(-1)^{x_{i}+y_{i}}\ket{i}\ket{i} \, .
\]

After both Alice and Bob perform a Hadamard transformation on their
qubits, we obtain
\[
\frac{1}{n \sqrt{n}}\sum_{i}\sum_{a,b}(-1)^{x_{i}+y_{i}}(-1)^{a\cdot
  i}(-1)^{b\cdot i}\ket{a}\ket{b} \, .
\]

So the probability $p_{a,b}$ of obtaining outcome $a,b$ when
measuring in the computational basis is 
\[
p_{a,b} = \frac{{1}}{n^3} \left| \sum_{i}(-1)^{x_{i}+y_{i}+(a+b)\cdot i} \right|^2
\]

If $x=y$, then $p_{a,b}=0$ wherever $a\neq b$.  If $\Delta(x,y)=n/2$,
then $p_{a,b}=0$ wherever $a=b$.  It follows that $EQ'(x,y)=EQ(a,b)$
--- determining the equality of the $n$-bit strings $x$ and $y$ is
equivalent to computing the equality of the $\log(n)$-bit strings $a$
and $b$. The garden-hose protocol for equality needs a number of pipes
that is linear in the input size.  After the quantum steps above,
Alice and Bob can use $O(\log
{n})$ water pipes to compute $EQ(a,b)$. \qed
\end{proof}

We can also show that the deterministic classical garden-hose complexity has an almost-linear lower bound.
\begin{theorem}
 $\gh(EQ') \in \Omega(\frac{n}{\log n})$
\end{theorem}
\begin{proof}
Theorem 1.7 of \cite{BCW98} shows that the zero-error classical
communication complexity of $EQ'$ is lower bounded by $\Omega(n)$. The
statement then follows from Proposition~\ref{prop:commtogh}. \qed
\end{proof}

\subsection{Randomized Setting}
The Noisy Perfect Matching problem (NPM) is a variant of the Boolean Hidden
Matching introduced in~\cite{GKKRdW07} where they prove an exponential
gap between the classical one-way communication complexity and the
quantum one-way communication complexity of NPM. We adapt the given
quantum one-way protocol to our setting, showing that the quantum
garden-hose complexity is only logarithmic.  This gives a
separation between the classical and quantum garden-hose complexity of
a partial function in the randomized setting.

The NPM problem is described as follows:\footnote{For this example,
  we deviate from the earlier convention of giving two $n$-bit strings as
  input to the players.}
\begin{description}
\item[Alice's input:] $x \in \set{0,1}^{2n}$.
\item[Bob's input:] a perfect matching $M$ on $\set{1,\ldots,2n}$ and a string $w
  \in \set{0,1}^n$.  The matching $M$ consists of $n$ edges, $e_1 =
  (i_1,j_1), \ldots, e_n = (i_n, j_n)$.
\item[Promise:] $\exists b\in \set{0,1}$ such that $\Delta(M\cdot x
  \, \oplus \, b^n, w) \leq n/3$, where $\Delta(\cdot,\cdot)$ is the Hamming
  distance and the $k$-th bit of the $n$-bit string $M\cdot x$ equals $x_{i_k} \oplus x_{j_k}$.

\item[Function value:] $b$.
\end{description}
Informally, the question asked is whether the parity on the edges of
$M$, where the vertices are entries of $x$, is close to the parities
specified by $w$, or not.

\begin{theorem}
 $\ghq(\mathrm{NPM}) \in O(\log n)$.
\end{theorem}
\begin{proof}
  Alice and Bob use $\log(2n)$ EPR pairs as quantum state $\ket{\psi}
  = \frac{1}{\sqrt{2n}}\sum_{i=0}^{2n-1}{\ket{i}\ket{i}}$.  Alice
  inserts her input bits $x=x_0 \ldots x_{2n-1}$ as phases of the
  shared superposition, yielding the shared state
\[
 \frac{1}{\sqrt{2n}}\sum_{i=0}^{2n-1}{(-1)^{x_i}\ket{i}_A\ket{i}_B} \, .
\]
Bob performs the following measurement: he uses projectors $P_k =
\proj{i_k}_B + \proj{j_k}_B$ corresponding to the $n$
edges. As they form a perfect matching, we have
$\sum_{k=1}^n P_k = I$ and $P_k P_{k'} = \delta_{k{k'}} P_k$, so
$\set{P_k}_k$ is a valid orthogonal measurement. Let us denote Bob's
measurement outcome by $\ell$. Setting $i:=i_\ell$ and $j:=j_\ell$,
the post-measurement state is
\[
 (-1)^{x_{i}} \ket{i}_A\ket{i}_B + (-1)^{x_{j}} \ket{j}_A\ket{j}_B \, .
\]
Alice then performs a Hadamard transform $H^{\otimes 2n} \otimes I$ on
her part of the state, resulting in
\[
 \sum_{a=0}^{2n-1}{\ket{a}_A \left[ (-1)^{x_i+a \cdot i}\ket{i}_B +
     (-1)^{x_j + a \cdot j}\ket{j}_B \right]} \, .
\]
Alice measures her register in the computational basis and obtains
outcome $a$. Bob performs a Hadamard gate on basis states $\ket{i}_B$
and $\ket{j}_B$, that is, $H_{i,j} = \frac12 (
  \proj{i}_B+\proj{i}_B+\proj{j}_B-\proj{j}_B )$, resulting in
the state
\begin{align*}
\ket{a}_A & \left( \frac12 \left[ (-1)^{x_i+a \cdot i} + (-1)^{x_j + a \cdot j}
\right] \ket{i}_B \right. \\
& \,\, \left. + \frac 12 \left[ (-1)^{x_i+a \cdot i} - (-1)^{x_j + a \cdot j} \right] \ket{j}_B \right)
\, .
\end{align*}
and measures in the computational basis. 
% Let $z_\ell = x_i \oplus x_j$.
He gets outcome $i$ if and only if $x_i \oplus a\cdot i = x_j \oplus a
\cdot j$ which is equivalent to $x_i \oplus x_j = a \cdot (i \oplus
j)$. In case $x_i \oplus x_j \neq a \cdot (i \oplus j)$, Bob gets
outcome $j$.

In the garden-hose game played after the measurements, Alice and Bob
perform the garden-hose protocol for the inner-product function
described in~\cite{Speelman11} with $a$ and $i \oplus j$ as their
respective inputs. The protocol can be easily adapted so that at the
end of it, the water will be in one particular pipe (known to Bob) on
Bob's side if $a \cdot (i \oplus j) = 0$, let us call this pipe
$0$-pipe. The water will be in another ``$1$-pipe'' (known to Bob) if
$a \cdot (i \oplus j) = 1$. Furthermore, Bob knows from his second
measurement outcome if they are computing $x_i \oplus x_j$ or $x_i
\oplus x_j \oplus 1$. In the first case, Bob looks at the $\ell$-th
bit of $w$ and leaves the $0$-pipe open if $w_\ell=1$ and routes the
$1$-pipe to Alice, and if $w_\ell=0$ he keeps the $1$-pipe open and
sends back the $0$-pipe. This strategy computes the function value
$w_\ell \oplus x_i \oplus x_j$, with $\ell$ uniformly random in
$\set{1,\ldots,n}$.  The promise guarantees that it gives the correct
value $b$ with probability at least $\frac{2}{3}$. The second case
(when Bob knows that $a \cdot (i \oplus j) \neq x_i \oplus x_j$) is
handled by the ``inverse'' strategy.  \qed
\end{proof}

\begin{theorem}
 $\ghr(\mathrm{NPM}) \in \Omega(\frac{\sqrt{n}}{\log n})$.
\end{theorem}
\begin{proof}
Combining the lower bound on the classical one-way communication complexity
from~\cite{GKKRdW07} of $\Omega(\sqrt{n})$ with
Proposition~\ref{prop:commtogh_rand} gives the statement. \qed
\end{proof}

%############################################################################
\section{Lower Bounds on Quantum Resources for a Perfect Attack} \label{app:lowerbound}
%############################################################################
We show that for a function that is injective for
Alice or injective for Bob (according to
Definition~\ref{def:injective}), the dimension of the state the
adversaries need to handle (including possible quantum communication
between them) in order to attack protocol $\PVqubit$
perfectly has to be of order at least linear in the classical input
size $n$. We start by showing two lemmas. The
actual bound is shown in Section~\ref{sec:actuallowerbound}.

In the last subsection, we show that there exist functions for which
perfect attacks on $\PVqubit$ requires the adversaries to handle a
polynomial amount of qubits.

\subsection{Localized Qubits} \label{sec:localized} Assume we have two
bipartite states $\ket{\psi^0}$ and $\ket{\psi^1}$ with the property
that $\ket{\psi^0}$ allows Alice to locally extract a qubit and
$\ket{\psi^1}$ allows Bob to locally extract the same
qubit. Intuitively, these two states have to be different. 

More formally, we assume that both states consist of five registers
$R,A,\tilde{A},B,\tilde{B}$ where registers $R,A,B$ are one-qubit
registers and $\tilde{A}$ and $\tilde{B}$ are arbitrary. We assume
that there exist local unitary transformations $U_{A \tilde{A}}$
acting on registers $A \tilde{A}$ and $V_{B \tilde{B}}$ acting on $B
\tilde{B}$ such that\footnote{We always assume that these
  transformations act as the identities on the registers we do not
  specify explicitly.}
\begin{align}
U_{A \tilde{A}} \ket{\psi^0}_{RA\tilde{A}B\tilde{B}} &=
\ket{\beta}_{RA} \otimes
\ket{P}_{\tilde{A}B\tilde{B}} \label{eq:unitaryAlice} \\
V_{B \tilde{B}} \ket{\psi^1}_{RA\tilde{A}B\tilde{B}} &=
\ket{\beta}_{RB} \otimes \ket{Q}_{A\tilde{A}\tilde{B}} \, ,
\label{eq:unitaryBob}
\end{align} 
where $\ket{\beta}_{RA} := (\ket{00}_{RA} + \ket{11})_{RA})/\sqrt{2}$
denotes an EPR pair on registers $RA$ and
$\ket{P}_{\tilde{A}B\tilde{B}}$ and $\ket{Q}_{A\tilde{A}\tilde{B}}$
are arbitrary pure states. 

\begin{lemma} \label{lem:innerproduct}
Let $\ket{\psi^0}, \ket{\psi^1}$ be states that
fulfill~\eqref{eq:unitaryAlice} and~\eqref{eq:unitaryBob}. Then,
\[ \big| \, \braket{\psi^0}{\psi^1} \, \big| \leq 1/2 \, . \]
\end{lemma}

\begin{proof}
  Multiplying both sides of~\eqref{eq:unitaryAlice} with
  $U_{A\tilde{A}}^\dag$ and multiplying~\eqref{eq:unitaryBob} with
  $V_{B\tilde{B}}^\dag$, we can write
\begin{align*}
  \big| \, \braket{\psi^0}{\psi^1} \, \big| &= \big| \,
  \bra{\beta}_{RA} \bra{P}_{\tilde{A}B\tilde{B}} \; U_{A\tilde{A}} \,
  V_{B\tilde{B}}^\dag \; \ket{\beta}_{RB} \ket{Q}_{A\tilde{A}\tilde{B}} \,
  \big| \\
  &= \big| \, \bra{\beta}_{RA} \bra{P'}_{\tilde{A}B\tilde{B}}
  \ket{\beta}_{RB} \ket{Q'}_{A\tilde{A}\tilde{B}} \, \big|\\
  &= \big| \, \bra{P'}_{\tilde{A}B\tilde{B}} \bra{\beta}_{RA}
  \ket{\beta}_{RB} \ket{Q'}_{A\tilde{A}\tilde{B}} \, \big| \; ,
\end{align*}
where we used that $U_{A\tilde{A}}$ and $V_{B\tilde{B}}$ commute and
defined $\ket{P'}_{\tilde{A}B\tilde{B}} := V_{B\tilde{B}}
\ket{P}_{\tilde{A}B\tilde{B}}$ and $\ket{Q'}_{A\tilde{A}\tilde{B}} := U_{A\tilde{A}}
\ket{Q}_{A\tilde{A}\tilde{B}}$. The last equality is just rearranging
terms that act on different registers.

Note that writing out the partial inner product between
$\ket{\beta}_{RA}$ and $\ket{\beta}_{RB}$ gives 
$$\bra{\beta}_{RA} \ket{\beta}_{RB} = \frac{1}{2} \big(\bra{0}_A \ket{0}_B 
+ \bra{1}_A \ket{1}_B \big) \; ,$$
where the operator in the parenthesis ``transfers'' a qubit from register
$A$ to register $B$.
Hence, 
\begin{align*}
  \big| \, \braket{\psi^0}{\psi^1} \, \big| &= 
\big| \, 
\bra{P'}_{\tilde{A}B\tilde{B}} \frac{1}{2} \big( \bra{0}_A \ket{0}_B + \bra{1}_A \ket{1}_B \big)  \ket{Q'}_{A\tilde{A}\tilde{B}} \,
\big| \\
&= \frac{1}{2} \cdot \big| \, \bra{P'}_{\tilde{A}B\tilde{B}}
\ket{Q'}_{B\tilde{A}\tilde{B}} \, \big|\\
&\leq \frac{1}{2} \; ,
\end{align*}
where the last step follows from the fact that the inner product
between any two unit vectors on the same registers can be at most 1.
\qed \end{proof}

\subsection{Squeezing Many Vectors in a Small Space}\label{sec:squeeze}
For the sake of completeness, we reproduce here an argument similar
to~\cite[Section 4.5.4]{NC00} about covering the state space of
dimension $d$ with patches of radius $\eps$.

\begin{lemma} \label{lem:squeeze}
Let $\mathcal{B}$ be a set of $2^n$ distinct unit vectors in a complex
Hilbert space of dimension $d$, with pairwise absolute inner product
at most $1/2$. Then, the dimension $d$ has to be in $\operatorname\Omega(n)$.
\end{lemma}

\begin{proof}
  For any two vectors $\ket{v},\ket{w}$, we can rotate the space such
  that $\ket{v}=\ket{0}$ and $\ket{w}=\cos \theta\ket{0}+\sin
  \theta\ket{1}$ for two orthogonal vectors $\ket{0}$ and
  $\ket{1}$. The \emph{Euclidean distance} between $\ket{v}$ and
  $\ket{w}$ can be expressed as
\begin{align*}
\big| \, \ket{v} - \ket{w} \, \big| &= | (1-\cos \theta) \ket{0} - \sin\theta
\ket{1} | \\
&= \sqrt{(1-\cos\theta)^2 + \sin^2\theta}\\
&= \sqrt{1- 2\cos\theta + \cos^2\theta + \sin^2\theta} \\
&= \sqrt{2} \sqrt{1 - \cos\theta} \, .
\end{align*}
If $\ket{v}$ and $\ket{w}$ have absolute inner product at most 1/2, we
have that $|\cos\theta|\leq 1/2$ and hence $\big| \, \ket{v} - \ket{w}
\, \big| \geq 1$. Therefore, the vectors in $\mathcal{B}$ have
pairwise Euclidean distance at least 1. The set of unit vectors
$\ket{w}$ with Euclidean distance at most $\delta$ from $\ket{v}$ is
called \emph{patch of radius $\delta$ around $\ket{v}$}.  It follows
that patches of radius $1/2$ around every vector in the set
$\mathcal{B}$ do not overlap.

The space of all $d$-dimensional state vectors can be
regarded as the real unit $(2d-1)$-sphere, because the
vector has $d$ complex amplitudes and hence $2d$ real degrees
of freedom with the restriction that the sum of the squared amplitudes
is equal to 1. Notice that the Euclidean distance between complex
vectors $\ket{v}$, $\ket{w}$ remains unchanged if we regard these
vectors as points of the real unit $(2d-1)$-sphere.

The surface area of a patch of radius $1/2$ near any vector is lower
bounded by the volume of a $(2d-2)$-sphere of radius $\eps$ where
$\eps$ is a constant slightly less than 1/2.\footnote{The patch is
  a ``bent'' version of this volume.}.  We use the formula $S_k(r)=2
\pi^{(k+1)/2} r^k/\operatorname\Gamma((k+1)/2)$ for the surface area
of a $k$-sphere of radius $r$, and $V_k(r)=2 \pi^{(k+1)/2}
r^{k+1}/[(k+1)\operatorname\Gamma((k+1)/2)]$ for the volume of a
$k$-sphere of radius $r$. The total surface area of all patches, which
is at least $2^n \cdot V_{2d-2}(\eps)$, is not more than the total
surface of the whole sphere $S_{2d-1}(1)$. Inserting the formulas, we
get
\begin{align*}
2^n \cdot 2 \pi^{d-\frac12}
\frac{\eps^{2d-1}}{(2d-1)\operatorname\Gamma(d-\frac12)}
 \leq 2 \pi^{d} \frac{1}{\operatorname\Gamma(d)}
\end{align*}
Using the fact that
$\frac{\operatorname\Gamma(d-\frac12)}{\operatorname\Gamma(d)} \leq
\frac{1}{d}$, we conclude that 
\begin{align*}
2^n \leq \sqrt{\pi} (2-\frac{1}{d}) \eps^{-(2d-1)} \leq 2 \sqrt{\pi}
\eps^{-(2d-1)} \, .
\end{align*}
As $\eps < 1/2$, we obtain that $d$ has to be in $\operatorname\Omega(n)$.
\qed \end{proof}

\subsection{The Lower Bound}\label{sec:actuallowerbound}
We consider perfect attacks on protocol $\PVqubit$ from
Figure~\ref{fig:PVqubit}. We allow the players one round of
simultaneous quantum communication which we model as follows.  Let
$\ket{\psi}_{RA\tilde{A} A_C B\tilde{B} B_C}$ be the pure state after
Alice received the EPR half from the verifier. The one-qubit register
$R$ holds the verifier's half of the EPR-pair, the one-qubit register
$A$ contains Alice's other half of the EPR-pair, the register
$\tilde{A}$ is Alice's part of the pre-shared entangled state and the
register $A_C$ holds the qubits that will be communicated to Bob. The
registers $B\tilde{B}B_C$ belong to Bob where $B$ holds one qubit and
$\tilde{B}$ is Bob's part of the entangled state and the $B_C$
register will be sent to Alice. We denote by $q_A$ the total number of
qubits in registers $\tilde{A}$ and $A_C$ and by $q_B$ the total
number of qubits in $\tilde{B}$ and $B_C$.  The overall state is thus
a unit vector in a complex Hilbert space of dimension $d :=
2^{2+q_A+1+q_B}$.

In the first step of their attack, Alice performs a unitary transform
$U^x$ depending on her classical input $x$ on her registers
$A\tilde{A}A_C$. Similarly, Bob performs a unitary transform $V^y$
depending on $y$ on registers $B\tilde{B}B_C$. After the application
of these transforms, the communication registers $A_C$ and $B_C$ and
the classical inputs $x$ and $y$ are exchanged. A final unitary
transform (performed either by Alice or Bob) depending on both $x,y$
``unveils'' the qubit either in Alice's register $A$ or in Bob's
register $B$.

\begin{theorem}
  Let $f$ be injective for Bob. Assume that Alice and Bob perform a
  perfect attack on protocol $\PVqubit$. Then, the dimension $d$ of
  the overall state (including the quantum communication) is in
  $\operatorname\Omega(n)$.
\end{theorem}

\begin{proof}
We assume that the player's actions are unitary transforms as
described before the theorem. % \chris{how
  % much space do we actually lose by assuming this unitarity?}

We investigate the set $\mathcal{B}$ of overall states after Bob performed his
operation, but \emph{before} Alice acts on the state. These states
 depend on Bob's input $y \in \set{0,1}^n$,
\begin{align*}
\mathcal{B} := \Set{ V_{B\tilde{B}B_C}^y \ket{\psi}_{RA\tilde{A}A_CB\tilde{B}B_C} }{y \in
  \set{0,1}^n} \, .
\end{align*}
We claim that for any two different $n$-bit strings $y \neq y'$, the
corresponding two vectors $V^y \ket{\psi}$ and $V^{y'} \ket{\psi}$ in
$\mathcal{B}$ have an absolute inner product of at most $1/2$.

Due to the injectivity of $f$, there exists an input $x$ for Alice such that $f(x,y) \neq
f(x,y')$. Applying Alice's unitary transform $U^x$ to both vectors
does not change their inner product, i.e.
\begin{align*}
| \bra{\psi} (V^y)^\dag V^{y'} \ket{\psi} | = | \bra{\psi}
(V^y)^\dag (U^x)^\dag U^x V^{y'} \ket{\psi} | \, .
\end{align*}
As $f(x,y) \neq f(x,y')$, the qubit has to end up on different
sides. Formally, there exist unitary transforms $K_{A \tilde{A} B_C}$
and $L_{B \tilde{B} A_C}$ that ``unveil'' the qubit in register $A$ or
$B$ respectively. Hence, we can apply Lemma~\ref{lem:innerproduct} to
prove the claim that the two vectors $V^y \ket{\psi}$ and $V^{y'}
\ket{\psi}$ have an absolute inner product of at most $1/2$. In
particular, all of the vectors in $\mathcal{B}$ are distinct. Applying
Lemma~\ref{lem:squeeze} yields the theorem.  \qed \end{proof}

\subsection{Functions For Which Perfect Attacks Need a Large Space}
Using similar arguments as above, we can show the existence of
functions for which perfect attacks require polynomially many
qubits.
\begin{theorem}
  For any starting state $\ket{\psi}$ of dimension $d$, there exists a
  Boolean function on inputs $x,y \in \set{0,1}^n$ such that any perfect
  attack on $\PVqubit$ requires $d$ to be exponential in $n$.
\end{theorem}

We believe that the statement with the reversed order of quantifiers
is true as well (but our current proof does not suffice for this
purpose), so that we can guarantee the existence of one particular
function (independent of the starting state) for which perfect attacks
require large states.
\begin{proof}[sketch]
  We consider covering the sphere with $K$ patches of vectors whose
  pairwise absolute inner product is larger than $\frac{\sqrt{3}}{2}$
  (which corresponds to an Euclidean distance of $\eps = \sqrt{2}
  \sqrt{1+\sqrt{3}/2} \approx 0.52$).  This partitioning also induces
  a partitioning on all possible unitary operations of Alice and Bob.
  We say that two actions $A$ and $A'$ are in the same patch if they
  take the starting state $\ket{\psi}$ to the same patch. In other
  words, if two actions are in the same patch then
\[
\bigl| \bra{\psi}A'^\dagger A \ket{\psi} \bigr| \geq \frac{\sqrt{3}}{2} \text{.}
\]
\noindent \emph{Claim.}  Given two actions of Alice $A,A'$ coming from
the same patch $i$, and two actions of Bob $B,B'$ coming from the same
patch $j$, the inner product between $B A \ket{\psi}$ and $B' A'
\ket{\psi}$ has magnitude at least $\frac{1}{2}$.
\begin{proof}[of the claim]
  Since Alice and Bob act on different parts of the state, their
  actions commute.  Write $\ket{\psi_A} := A'^\dagger A \ket{\psi}$
  and $\ket{\psi_B} := B^\dagger B' \ket{\psi}$.  Then the inner
  product can be written as
\[
\bra{\psi}A'^\dagger B'^\dagger B A \ket{\psi} =  \bra{\psi}B'^\dagger B A'^\dagger A \ket{\psi} = \braket{\psi_B}{\psi_A}
\]
Note that
\[
\bigl| \braket{\psi}{\psi_A} \bigr| = \bigl| \bra{\psi} A'^\dagger A \ket{\psi} \bigr| \geq \frac{\sqrt{3}}{2}\text{,}
\] so the angle $\theta$ between $\ket{\psi_A}$ and $\ket{\psi}$ is at
most $\arccos{\frac{\sqrt{3}}{2}} = \frac{\pi}{6}$. The same holds for
the angle between $\ket{\psi_B}$ and $\ket{\psi}$.  We can upper bound
the total angle between $\ket{\psi_A}$ and $\ket{\psi_B}$ by the sum
of these angles, giving a total angle of at most $\frac{\pi}{3}$.
This corresponds to a lower bound on the inner product of
$\cos{\frac{\pi}{3}}=\frac{1}{2}$. \qed
\end{proof}

So there exists no pair of combined actions $AB$ and $A'B'$, with $A$
and $A'$ in patch $i$ and $B$ and $B'$ in patch $j$, such that the
qubit ends up on Alice's side for $AB$ and on Bob's side for
$A'B'$. Therefore, the combination of $i$ and $j$ completely
determines the destination of the qubit and hence the output of the
function. 
If $K$ denotes the number of patches, then there are $K^{2^n}$
possible strategies for Alice and $K^{2^n}$ possible strategies for
Bob. Hence, the number of combined strategies (possibly resulting in
different functions) is at most $K^{2\cdot2^{n}}$.

It is shown in~\cite[Section 4.5.4]{NC00} that we need
at least $K=\operatorname\Omega(\frac{1}{\eps^{d-1}})$ patches.  Using the
same counting argument as in Proposition~\ref{prop:expbound}, we have
that
\begin{align*}
2^{2^{2n}} \geq \operatorname\Omega \left(\frac{1}{\eps^{(d-1)2\cdot 2^n}}\right) \, ,
\end{align*}
from which follows that for some function, $d$ has to be exponential
in $n$.%  We conclude that there exists a function for which the
% players' state space has to contain a polynomial
% amount of qubits.
\qed \end{proof}

\bibliographystyle{alpha}
\bibliography{library}

\end{document}